\documentclass[12pt]{article}
 \usepackage{xcolor}
\usepackage{amsfonts,amsmath}
\usepackage[mathscr]{eucal}
\usepackage{amssymb}
\usepackage{bm}
\usepackage{bbold}
\usepackage{amsthm}
\usepackage{graphicx,subfigure}
\usepackage{dsfont}

\usepackage{ytableau}

\usepackage{overpic}

\usepackage{hyperref}

\usepackage[titletoc,toc,title]{appendix}

\newcommand{\appendixnumberline}[1]{Appendix #1.\space}

\let\oldappendix\appendix
\makeatletter
\renewcommand{\appendix}{%
  \addtocontents{toc}{\let\protect\numberline\protect\appendixnumberline}%
  \renewcommand{\@seccntformat}[1]{
  \bfseries Appendix \Alph{section}. }%
  \oldappendix
}
\makeatother

\theoremstyle{plain}

\newtheorem{prop}{Proposition}

\definecolor{navycol}{RGB}{100,150,160}
   \definecolor{pinkcol}{RGB}{242,55,55}
   \definecolor{bluecol}{RGB}{50,205,50}
   \definecolor{bluecol}{RGB}{30,144,255}
   \definecolor{yellowcol}{RGB}{255,252,134}
   \definecolor{lbluecol}{RGB}{214,252,168}
   
   \usepackage{empheq}

\def\theequation{\arabic{section}.\arabic{equation}}
\newcommand{\be}{\begin{eqnarray}}
\newcommand{\red}{\color{red}}
\newcommand{\blue}{\color{blue}}
\newcommand{\ee}{\end{eqnarray}}
\newcommand{\nn}{\nonumber \\}
\newcommand{\lb}{\label}
\newcommand{\p}[1]{(\ref{#1})}

\newcommand{\ga}{\lower.7ex\hbox{$
\;\stackrel{\textstyle>}{\sim}\;$}}
\newcommand{\vecg}[1]{\mbox{\boldmath $#1$}}
\newcommand{\la}{\lower.7ex\hbox{$
\;\stackrel{\textstyle<}{\sim}\;$}}

\newcommand{\pd}{\partial}
\newcommand{\vecind}[1]{\mbox{\scriptsize \boldmath $#1$}}

\usepackage{setspace}
\usepackage{tcolorbox}
\tcbuselibrary{theorems}

\newtcbtheorem
  []
  {theorem}
  {Theorem}
  {
    colback=bluecol!20!,
    colframe=blue!50!,
    fonttitle=\sffamily,
  }
  {def}

  \newtcbtheorem
  []
  {lemma}
  {Lemma}
  {
    colback=yellowcol!30!,
    colframe=pinkcol!60!,
    fonttitle=\sffamily,
  }
  {def}

\let\OLDthebibliography\thebibliography
\renewcommand\thebibliography[1]{
  \OLDthebibliography{#1}
  \setlength{\parskip}{4pt}
  \setlength{\itemsep}{0pt plus 0.3ex}
}

\begin{document}

\begin{titlepage}

\vspace*{0.2cm}

\renewcommand{\thefootnote}{\star}
\begin{center}

{\LARGE\bf  Cartan monopoles}\\

\vspace{0.5cm}

\vspace{1.5cm}
\renewcommand{\thefootnote}{$\star$}

\quad {\large\bf Andrei~Smilga} 
 \vspace{0.5cm}

{\it SUBATECH, Universit\'e de Nantes,}\\
{\it 4 rue Alfred Kastler, BP 20722, Nantes 44307, France;}\\
\vspace{0.1cm}

{\tt smilga@subatech.in2p3.fr}\\

\end{center}

\vspace{0.2cm} \vskip 0.6truecm \nopagebreak

\begin{abstract}
\noindent  
The effective Hamiltonians for chiral supersymmetric gauge theories at small spatial volume are generalizations of the Hamiltonians describing the motion of a scalar or a spinor particle in  a field of Dirac monopoles (we are dealing in fact with a certain lattice of monopoles supplemented with a periodic singular potential).  The gauge fields in such  Hamiltonians belong to the Cartan subalgebras of the corresponding gauge algebras. Such a  construction exists for all groups admitting complex representations, i.e. for $SU(N \geq 3), \ Spin(4n+2)$ with $n \geq 1$ and $E_6$. We give explicit expressions for these Hamiltonians  for $SU(3)$, $SU(4) \simeq Spin(6)$ and for $SU(5)$. The simplified version of such a Hamiltonian, deprived of fermion terms, of the extra scalar potential and when only one node of the lattice is taken into consideration, describe a $3r$-dimensional motion ($r$ being the rank of the group) in  the field  what we call a {\it Cartan monopole}. As is the case for the ordinary monopole, the Lagrangian of this system enjoys gauge symmetry, rotational symmetry, and the parameter, generalizing the notion of magnetic charge for Cartan monopoles, is quantized. 
 \end{abstract}

\vspace{1cm}
\bigskip

\newpage

\end{titlepage}

\tableofcontents

\setcounter{footnote}{0}

\setcounter{equation}0
\section{Dirac monopole and associated Hamiltonians}

We start with recalling  some basic facts concerning the dynamics of Dirac monopole \cite{Dirac}
The magnetic field has a Coulomb form:
\be
\lb{field-mon}
{\cal H}_i \ =\ \frac {qx_i}{r^3}\,,
\ee
where $q$ is the magnetic charge.

The corresponding gauge potentials may be chosen in the form
\be
\lb{A-mon}
{\cal A}_x \ =\ - \frac{qy}{r(r+z)}, \quad {\cal A}_y \ =\  \frac{qx}{r(r+z)}, \quad
{\cal A}_z \ =\ 0 \,.
\ee
The field \p{A-mon} is singular on the Dirac string, $\vecg{x} =  \vecg{n} r = (0,0,-r)$. Any other direction  $\vecg{n}$ of the string can be chosen. In particular, one can choose 
$\vecg{n} = (0,0,1)$,
 in which case the vector potentials
 \be
\lb{A-mon1}
\tilde {\cal A}_x \ =\ \frac{qy}{r(r-z)}, \quad \tilde {\cal A}_y \ =\  -\frac{qx}{r(r-z)}, \quad
\tilde {\cal A}_z \ =\ 0 
\ee
have the same curl \p{field-mon}.

The potentials \p{A-mon} and \p{A-mon1} are related by a gauge transformations
\be 
\lb{gauge-A}
 \tilde {\cal A}_j \ = \ {\cal A}_j - \pd_j \chi
 \ee 
with
\be
\chi(x,y,z) \ =\ 2q \arctan \frac yx \ =\ 2q\phi\,,
\ee
where $\phi$ is the azymuthal angle.

By considering the interaction of the monopole with a particle carrying an electric charge $e$, one derives the quantization condition,
 \be
\lb{quant-q}
qe \ =\ \frac n2
\ee
with integer $n$.\footnote{It is written in the unit system $\hbar = c = 1$. In the following, we will always set $e=1$ (with the positive sign, not the negative one as would be for a physical electron) so that $q = n/2$.}
There are two physical ways and a more mathematical way to derive \p{quant-q}.
\begin{itemize}
\item One can calculate the angular momentum of the electromagnetic field created by the monopole and the electric charge,
\be
\lb{J-field}
\vecg{J}^{\rm field}  \ =\ \frac 1{4\pi} \int d\vecg{x} \  [\vecg{x} \times [\vecg{E} \times \vecg{\cal H}]]\,,
\ee
 and require for it to be equal to $n/2$ with integer $n$.
\item One can consider the Schr\"odinger equation describing the motion of the charged particle in the monopole field,
\be
\lb{Schrod}
\hat H \Psi \ =\ \frac 12 (\hat P_j + {\cal A}_j)^2 \Psi \ =\ E\Psi 
 \ee
with $\hat P_j = -i \pd/ \pd x_j$. The gauge transformation \p{gauge-A} induces the corresponding transformation of the eigenfunctions of \p{Schrod}: $\tilde \Psi \ =\ e^{i\chi} \Psi$.
The requirement for $e^{i\chi}$ and hence for the wave function to be uniquely defined brings about the condition \p{quant-q}.\footnote{Well, one can {\it abandon} the idea of gauge invariance, pick up a particular direction of the Dirac string, in which case the spectral problem for the  Hamiltonian \p{Schrod} can in principle be formulated even for fractional magnetic charges. This means, however, that the Dirac string becomes {\it observable}, which is not what we want \cite{flux}.}

\item A more mathematical derivation was suggested in Ref. \cite{Wu-Yang}. 
Separating in \p{A-mon} the radial dependence, we arrive at an Abelian gauge field on $S^2$ that is singular at its south pole. Likewise, the field \p{A-mon1} is singular at its north pole. 

One now can construct a {\it fiber bundle} subdividing $S^2$ into two charts: the north hemisphere where the gauge field has the form 
\be
\lb{A-mon-north}
{\cal A}_x \ =\ - \frac{q \sin \theta \sin \phi}{r(1 + \cos \theta)}, \quad {\cal A}_y \ =\ 
\frac{q \sin \theta \cos \phi}{r(1 + \cos \theta)}, \quad
{\cal A}_z \ =\ 0 
 \ee
and the south hemisphere where 
\be
\lb{A-mon-south}
\tilde {\cal A}_x \ =\ \frac{q \sin \theta \sin\phi}{r(1 - \cos \theta)}, \quad \tilde {\cal A}_y \ =\ 
-\frac{q \sin \theta \cos \phi}{r(1 - \cos \theta)}, \quad
\tilde {\cal A}_z \ =\ 0 \,.
 \ee
 Two charts overlap on the equator, and if we require for $\tilde {\cal A}_j$ and ${\cal A}_j$ to be related there  by a {\it uniquely defined} gauge transformation, the quantization condition \p{quant-q} follows.

\end{itemize}

The Lagrangian, corresponding to the classical version of the quantum Hamiltonian in \p{Schrod}, reads
\be
\lb{L-mon}
L \ =\ \frac {\dot{x}_j  \dot{x}_j}2  - \dot{x}_j {\cal A}_j (\vecg{x})\,.
\ee

 \begin{prop}
The action $S$ of the system \p{L-mon} with $\vecg{\cal H} = \vecg{\nabla} \times \vecg{\cal A}$ given in Eq. \p{field-mon} is invariant under rotations. The conserved angular momentum is
\be
\lb{Mom-mon}
J_m \ =\ \varepsilon_{mnk} x_n \dot{x}_k + \frac {q x_m}r = 
\varepsilon_{mnk} x_n (P+ {\cal A})_k + \frac {q x_m}r\,. 
\ee
\end{prop}

\begin{proof}
It is quite explicit. The variation of the Lagrangian under $\delta x_j = \varepsilon_{jmn} \alpha_m x_n$ is
\be
\delta L \ =\ -\alpha_m \ \left(\varepsilon_{jmn} \dot{x}_n {\cal A}_j + \dot{x}_j \frac {\pd {\cal A}_j}{\pd x_p} \varepsilon_{pmn} x_n   \right)\,.
 \ee
After some massaging --- flipping  the time derivative in the first term and renaming the indices in the second term --- this can be represented as 
\be
\lb{tot-der}
\delta L \ =\ -\alpha_m \frac d{dt} \left( \varepsilon_{jmn} x_n {\cal A}_j \right) 
+\alpha_m (x_p {\cal H}_m - \delta_{mp} x_n {\cal H}_n ) \dot{x}_p = \nn
-\alpha_m  \frac d{dt}  \left(\varepsilon_{jmn} x_n {\cal A}_j  + \frac {qx_m}r    \right),
\ee
so that $\int L dt$ is invariant.

The conserved angular momentum can be derived by a standard Noether method.\footnote{The simplest way to proceed is first to find $J_3$ by evaluating the variation of the Lagrangian with the vector potentials \p{A-mon-north} under time-dependent rotations around the third axis. From this, the covariant expression \p{Mom-mon} can be easily restored.}
\end{proof}
The term $qx_m/r$ in $J_m$ is nothing but  the angular momentum \p{J-field}  of the  electromagnetic field.

 {\bf Remark.} The fact that the Lagrangian \p{L-mon} enjoys rotational invariance is not completely trivial. It would be if the potential
$$ \vecg{\cal A} \ =\ \frac {q \vecg{x} \times \vecg{n}}{r(r - \vecg{x} \cdot \vecg{n})} 
$$
 transformed as a vector under rotations of coordinates. But it {\it does} not. ${\cal A}_j$ transform as vector components only if one simultaneously rotates the direction $\vecg{n}$ of the Dirac string. But such a transformation acting not only on the dynamic variables $x_j$, but also on the parameters $n_j$, does not have Noether nature and the invariance under such transformation does not directly entail the existence of an integral of motion.

\vspace{1mm}

The Schr\"odinger equation \p{Schrod} with the monopole vector potentials  was solved by Tamm
\cite{Tamm}.  The wave functions are expressed via Jacobi polynomials. The allowed eigenvalues of $ \hat{\vecg{J}}^2$ are $j(j+1)$ with $j = |q|, |q|+1, \ldots $. The spectrum is
  \be
\lb{spec-mon}
E_{L=0,1,\ldots} \ =\ j(j+1) - q^2 
 \ee
with $2j + 1$ degenerate states on each level.

For a spin $\frac 12$ particle, the quantum Hamiltonian acquires a matrix form. We have to add to \p{Schrod}  the term $\frac\gamma 2{\cal H}_i \sigma_i$, where $\gamma$ is the gyromagnetic ratio of the particle. This  Schr\"odinger problem was also solved \cite{Malkus}. 

\section{Effective Hamiltonian for the chiral supersymmetric QED}
\setcounter{equation}0

We will show here how the Hamiltonian describing the interaction of charged particles with magnetic monopole appears naturally as an effective Hamiltonian for a {\it chiral} Abelian supersymmetric gauge theory placed in a small spatial box \cite{ja-Ab} (see also the review \cite{obzor} and Chapter 8 of the book \cite{Wit-book}). 

The simplest  such theory includes a vector multiplet $V$, eight  left-handed chiral matter superfields  carrying charge\footnote{It is a physical electric charge (as far as the fields in a hypothetical supersymmetric theory may be called physical), not the charge $e$  that entered Eq. \p{quant-q} and that we  set to 1. The monopoles that we are interested in in this paper enter the effective Hamiltonians living  in unphysical field spaces.}  $e$ and a right-handed superfield with charge $2e$.
In this case,  in contrast to the ordinary QED or SQED, the theory has no chiral left-right symmetry, but it is still anomaly-free and renormalizable due to the fact that the sums of the cubes of the charges for the left-handed and right-handed fermions are the same.

It is convenient to describe it only in terms of left chiral multiplets with eight multiplets $S^f$ carrying a unit charge $e$ and a multiplet $T$ carrying the charge $-2e$. Then
\be
\lb{sum-cubes}
\sum_k q_k^3 \ =\ 0\,,
 \ee
where the sum runs over all the multiplets. 
The Lagrangian of the model reads

\begin{empheq}[box = \fcolorbox{black}{white}]{align}
\lb{L-SQED-chiral}
{\cal L}\ = \frac 14 \int \!\! d^4\theta\, (\overline S^f e^{eV} S^f + \overline T e^{-2eV} T) + 
\left( \frac 18 \int  \!\! d^2\theta \,W^\alpha W_\alpha   + {\rm c.c.} \right).
 \end{empheq}
We put this theory in a finite spatial box of length $L$ and impose periodic boundary conditions. If $|e| \ll 1$, which we assume, one can   proceed in the spirit of \cite{Wit82} and subdivide all the dynamical variables in two classes: 

$\bullet$ {\it slow} variables, which are the  zero Fourier harmonics of the gauge field 
$c_j = A_j^{(\vecind{0})}$ and their superpartners $\eta_\alpha = \lambda_\alpha^{(\vecind{0})}$ and

$\bullet$ {\it fast} variables --- all the rest. 

The theory is invariant under gauge transformations, which include also {\it large} (topologically nontrivial) gauge transformations\footnote{In principle, one could relax the invariance requirement and impose a generic superselection rule such that all the wave functionals $\Psi[A, s^f, \psi^f, i, \xi]$ are multiplied by a phase factor exp$\{i \theta_j n_j\}$ with some particular $n_j$ after the transformation  \p{large}. But to keep supersymmetry, we have to stay in the  sector with zero  angle $\vecg{\theta} = \vecg{0}$.}
\be
\lb{large}
&&A_j \to A_j + \frac{2\pi n_j}{eL}, \nn
 &&(s^f,\psi^f)  \to (s^f,\psi^f) e^{-2\pi i \vecind{n} \cdot \vecind{x}/L}, \quad (t, \xi) \to (t,\xi) e^{4\pi i \vecind{n} \cdot \vecind{x}/L}\,.
 \ee
Then the zero modes of $A_j(\vecg{x})$ lie on a dual torus, $c_j \in [0, 2\pi/(|e|L)]$.

Bearing this in mind, the characteric excitation energy of slow degrees of freedom is $E_{\rm slow} \sim e^2/L$, which is much less than the characteric excitation energy of fast variables, $E_{\rm fast} \sim 1/L$. The latter can be integrated over to derive the effective Hamiltonian including only slow variables.

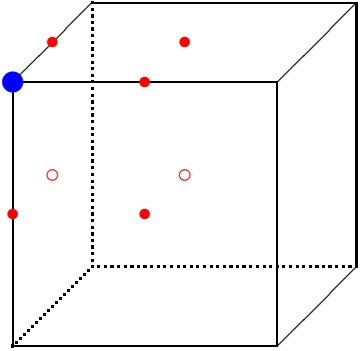
\begin{figure}

\begin{picture}(150,150)
 \put(100,0){\line(1,0){100}}
\put(100,0){\line(0,1){100}}
\put(100,100){\line(1,0){100}}
\put(200,0){\line(0,1){100}}
\put(100,100){\line(1,1){30}}
\put(200,100){\line(1,1){30}}
\put(130,130){\line(1,0){100}}
\put(200,0){\line(1,1){30}}
\put(230,30){\line(0,1){100}}

\put(100,100){\blue \circle*{8}}
\put(100,50){\red \circle*{4}}
\put(150,100){\red \circle*{4}}
\put(150,50){\red \circle*{4}}
\put(115,115){\red \circle*{4}}
\put(165,65){\red \circle{4}}
\put(165,115){\red \circle*{4}}
\put(115,65){\red \circle{4}}

\thicklines
\qbezier[20](100,0)(115,15)(130,30)
\qbezier[40](130,30)(180,30)(230,30)
\qbezier[40](130,30)(130,80)(130,130)

\end{picture}

 \caption{A unit cubic cell of the \index{magnetic!monopole} monopole crystal in field space. The edge of the cube is $2\pi/(|e|L)$.} 
\label{crystal}

 \end{figure} 

In the lowest Born-Oppenheimer approximation \cite{BO}, the effective supercharges and the Hamiltonian read \cite{ja-Ab}
\be
\lb{Q-chiral-SQED}
   \hat Q_\alpha^{\rm eff} \ =\ \left[(\sigma_k)_\alpha{}^\beta (\hat P_k + {\cal A}_k) +i  \delta_\alpha^\beta {\cal D}    \right] \eta_\beta, \nn
\hat {\bar Q}^{\alpha \, {\rm eff}} \ =\  \bar \eta^\beta \left[(\sigma_k)_\beta{}^\alpha (\hat P_k + {\cal A}_k) - i \delta_\beta^\alpha {\cal D}    \right], 
 \ee
\begin{empheq}[box = \fcolorbox{black}{white}]{align}
\lb{H-chiral}
\hat H^{\rm eff} \ =\ \frac 12 (\hat P_k + {\cal A}_k )^2 + \frac 12 {\cal D}^2 + \vecg{\cal H} \!\cdot \!\bar \eta \vecg{\sigma} \eta\,,
  \end{empheq}
where
\be
\lb{D-SQED}
{\cal D}(\vecg{c}) \ =\ 4\sum_{\vecind{n}} \frac 1 {\left|\vecg{c} - \frac {2\pi \vecind{n}}{eL} \right|}  - \frac 12 \sum_{\vecind{n}}   \frac 1{\left|\vecg{c} - \frac {\pi \vecind{n}}{eL} \right| } .      
 \ee 
and $\vecg{\cal A}(\vecg{c})$ is related to ${\cal D}(\vecg{c}) $  according to
\be
\lb{cond-KAH}
\vecg{\cal H} \ =\ \vecg{\nabla} \times \vecg{\cal A} \ =\ \vecg{\nabla} {\cal D} \,. 
\ee
The algebra
\be
\lb{algN2}
\{\hat Q_\alpha^{\rm eff}, \hat {\bar Q}^{\beta \, {\rm eff}}\} \ =\ 2\delta_\alpha^\beta 
\hat H^{\rm eff}
\ee
holds. 
The system belongs to a class of ${\cal N} = 2$ supersymmetric quantum mechanical systems constructed in 
\cite{Ritten}.

The function \p{D-SQED} is periodic in $\vecg{c}$ with the period $2\pi/(eL)$. We see a kind of crystal shown in Fig. \ref{crystal}.
In each node of this crystal, sits a monopole. The {\blue blue} blob marks the site with the monopole of charge $-7/2$. 
The {\red red} blobs mark the sites with the monopoles of charge $1/2$. Note that the net magnetic charge of the elementary cell is zero --- that is a consequence of the condition
\p{sum-cubes}.

Consider the Hamiltonian  \p{H-chiral} in the vicinity of one of the red nodes. It acts on the wave functions with fermion charges $F=0,1,2$. In the bosonic sectors, the Hamiltonian includes a monopole term like in \p{Schrod}  and an extra singular scalar potential. In the sector $F=1$, also the last term in \p{H-chiral} works. One may represent the wave function
$\Psi^{F=1}(\vecg{c}, \eta_\alpha) = \Psi_\alpha(\vecg{c}) \eta_\alpha$ as a spinor, so that  the  term $\vecg{\cal H} \!\cdot \! \bar \eta \vecg{\sigma} \eta$  acquires the form $\vecg{\cal H} \!\cdot \!\vecg{\sigma} $ and describes the interaction with the magnetic moment of the spin $\frac 12$ particle with the gyromagnetic ratio $\gamma = 2$ (as for the electron in the leading in $\alpha$ approximation).

The equation for the radial wave function $\chi(r)$ in the vicinity of the red node is 
\be
-\frac 1{2r} \frac {\pd^2}{(\pd r)^2}(r\chi) - \frac 1{8r^2} \chi \ =\ E\chi\,,
 \ee
where $r$ is the distance from the node.
This equation has a zero energy solution, $\chi(r) \sim 1/\sqrt{r}$. The normalization integral $\int r^2 \chi(r) dr$ for this solution diverges, but one can deform the Hamiltonian \p{H-chiral} adding a positive constant $C$ to  ${\cal D}(\vecg{c})$, so that the condition \p{cond-KAH} and the algebra \p{algN2} still hold. Then the zero energy solution of the corresponding radial equation would include an exponenial factor $e^{-Cr}$ and the wave function becomes normalizable. Moreover, for large enough $C$, the zero mode that we found will be concentrated in the region around the chosen  {\red red} node and will not sense the precence of other nodes. As the elementary cell involves 7 red monopoles with a positive magnetic charge and as, for positive $C$, there is no zero-energy solution around a {\blue blue} negatively charged monopole, we can derive the existence of seven zero energy solutions --- not only for the deformed, but also for the undeformed effective Hamiltonian and hence \cite{Wit82} also for the complicated field Hamiltonian of the original theory \p{L-SQED-chiral}. All these states are fermionic and hence the Witten index $I_W$ of chiral SQED in \p{L-SQED-chiral} is equal to \cite{ja-Ab,Wit-book}
 \be
\lb{IW-Ab}
I_W \ =\ -7\,.
  \ee
One can also  invent more complicated anomaly-free chiral SQED theories. Take for example the theory involving 27 left multiplets  with charge $e$ and a left multiplet of charge $-3e$. In that case, the Witten index is equal to $-27$.

\section{$SU(3)$ theory}
\setcounter{equation}0

These constructions can be generalized to non-Abelian theories. The group $SU(2)$ has only self-conjugate representations, and the corresponding theories are not chiral. The simplest chiral non-Abelian supersymmetric gauge theory is based on  $SU(3)$. It includes, besides the vector multiplet,  seven matter chiral multiplets in the fundamental group representation $\bf 3$ and a multiplet in the representation $\bar {\bf 6}$. \index{representation!sextet of $SU(3)$}
In terms of superfields, the Lagrangian of the model is
 \begin{empheq}[box = \fcolorbox{black}{white}]{align}
\lb{chiral-SU3}
{\cal L} \ =\ \left( \frac 14 {\rm Tr} \int d^2\theta\, \hat W^\alpha \hat W_\alpha  + {\rm c.c.} \right)
 + \frac  14 \int d^4\theta \left(\sum_{f=1}^7 \overline S^{fj} e^{g\hat V} S^f_j + 
\Phi^{jk} e^{-g\hat V  } \overline \Phi_{jk} \right) 
 \end{empheq}
with $j,k = 1,2,3$ and $\Phi^{jk} = \Phi^{kj}$. 

The generators of the group in the sextet representation read
\be
(T^a)_{ij}^{kl} \ =\ \frac 12 \left[ (t^a)_i^k \delta_j^l + (t^a)_i^l \delta_j^k + (t^a)_j^k \delta_i^l + (t^a)_j^l \delta_i^k \right]\,,
\ee
where $t^a$ are the generators in the fundamental representation and the factor $1/2$ stems from the requirement that $T^a$ satisfy the same commutation relations, $[T^a, T^b] = i f^{abc} T^c$, as $t^a$. Then 
\be
\lb{dabc-SU3}
 {\rm Tr} \{ T^{(a} T^b T^{c)} \} \ = \ 7 \,{\rm Tr} \{t^{(a} t^b t^{c)} \} \ =\ \frac 74 d^{abc}. 
 \ee
The contribution of each matter multiplet to the chiral anomaly is proportional to $ {\rm Tr} \{ T^{(a} T^b T^{c)} \}$ for the generators in the corresponding representation. For the theory \p{chiral-SU3}, the anomaly {\it cancels} and the theory is gauge-invariant and renormalizable.

 The effective Hamiltonian at small finite box can be found along the same lines as for the Abelian theory. We have 6 bosonic slow variables --- zero Fourier modes of the gauge fields belonging to the Cartan subalgebra of $SU(3)$ which can be represented as
 \be
\lb{Cartan}
\vecg{A}^a t^a \ =\ \frac 12 \,{\rm diag} (\vecg{a}, \vecg{b} - \vecg{a}, - \vecg{b})
\ee
with  $a_j, b_j \in \left[0, \frac{4\pi}{gL} \right)$ (the shifts  of $a_j, b_j$ by $4\pi/(gL)$ boil down to topologically trivial gauge transformations).
 The Hamiltonian also includes their fermion superpartners 
\be
\eta_\alpha^{(a)} \ &=&\ \eta_\alpha^3 + \frac {\eta_\alpha^8}{\sqrt{3}}\,, \nn
\eta_\alpha^{(b)} \ &=&\  \frac {2\eta_\alpha^8}{\sqrt{3}} 
\ee
and their conjugates. They have the following anticommutators:
\be
\lb{anticom-eta}
 \{\hat{\bar \eta}^{(a)\beta}, \eta^{(a)}_\alpha\} = \{\hat{\bar \eta}^{(b)\beta}, \eta^{(b)}_\alpha\} = \frac 43 \delta_\alpha^\beta,
 \quad \{\hat{\bar \eta}^{(a)\beta}, \eta^{(b)}_\alpha\} = \{\hat{\bar \eta}^{(b)\beta}, \eta^{(a)}_\alpha\} = \frac 23 \delta_\alpha^\beta\,, \nn
 \ee

In contrast to the pure SYM theory where the effective Hamiltonian to the lowest BO order describes free motion  in $\vecg{a}, \vecg{b}$ space
\cite{Wit82}, in the chiral case, the effective supercharges and Hamiltonian have a rather nontrivial form. One can derive  \cite{Blok}\footnote{This derivation for $SU(3)$ is spelled out in detail in the book \cite{Wit-book}. In the Appendix, we will describe a similar derivation for $SU(4)$.} 
 \be
\lb{Q-chiral}
   \hat Q_\alpha^{\rm eff} \ &=& \  {\eta^{(a)}_\beta} \left[(\sigma_k)_\alpha{}^\beta (\hat P_k^{(a)} + {\cal A}^{(a)}_k) 
+i  \delta_\alpha^\beta {\cal D}_{(a)}    \right] \nn
&+& \   {\eta^{(b)}_\beta} \left[(\sigma_k)_\alpha{}^\beta (\hat P_k^{(b)} + {\cal A}^{(b)}_k) +i  \delta_\alpha^\beta {\cal D}_{(b)}    \right],  \nn
\hat {\bar Q}^{\alpha \, {\rm eff}} \ &=&\    {\hat{\bar \eta}^{(a)\beta}}\left[(\sigma_k)_\beta{}^\alpha (\hat P_k^{(a)} 
 + {\cal A}^{(a)}_k) - i \delta_\beta^\alpha {\cal D}_{(a)}    \right] \nn
&+&\ {\hat{\bar \eta}^{(b)\beta}}  \left[(\sigma_k)_\beta{}^\alpha (\hat P_k^{(b)} + {\cal A}^{(b)}_k) - i \delta_\beta^\alpha {\cal D}_{(b)}    \right],
 \ee
and 
 \begin{empheq}[box = \fcolorbox{black}{white}]{align}
\lb{H-chiral-nab}
&\hat H^{\rm eff} =  \frac 23 \left[(\hat P_k^{(a)} + {\cal A}^{(a)}_k )^2 + (\hat P_k^{(b)} + {\cal A}^{(b)}_k )^2 + (\hat P_k^{(a)} + 
{\cal A}^{(a)}_k ) (\hat P_k^{(b)} + {\cal A}^{(b)}_k) \right. \nn 
&+ \left. {\cal D}_{(a)}^2 + {\cal D}_{(b)}^2 
   + {\cal D}_{(a)} {\cal D}_{(b)}\right]\nn 
 &+ \, \vecg{\cal H}^{(a)}\, \hat{\bar \eta}^{(a)} \vecg{\sigma} \eta^{(a)} +  \vecg{\cal H}^{(b)}\, \hat{\bar \eta}^{(b)} \vecg{\sigma} \eta^{(b)} +  \vecg{\cal H}^{(ab)}( \hat{\bar \eta}^{(a)} \vecg{\sigma} \eta^{(b)} 
+  \hat{\bar \eta}^b \vecg{\sigma} \eta^{(a)} )\,.
  \end{empheq}
Here $\hat P_k^{(a)} = -i \pd/\pd a^k$ and $\hat P_k^{(b)} = -i \pd/\pd b^k$. The functions ${\cal D}_{(a)}$ and ${\cal D}_{(a)}$ represent the following infinite sums:
\be
\lb{Dab}
\!\!\!\!\!{\cal D}_{(a)} &=& \frac 92 \sum_{\vecind{n}}\left[\frac 1{|\vecg{a} - 4\pi \vecg{n}|} - \frac 1{|\vecg{a} - \vecg{b} - 4\pi \vecg{n}|}   \right] \ - \  \frac 12     \sum_{\vecind{n}} \left[\frac 1{|\vecg{a} - 2\pi \vecg{n}|} - \frac 1{|\vecg{a} - \vecg{b} - 2\pi \vecg{n}|} \right] , \nn
\!\!\!\!\!{\cal D}_{(b)} &=& \frac 92 \sum_{\vecind{n}}\left[ \frac 1{|\vecg{a} - \vecg{b} - 4\pi \vecg{n}|}  - \frac 1{|\vecg{b} - 4\pi \vecg{n}|}\right] \ - \ \frac 12 \sum_{\vecind{n}}\left[ \frac 1{|\vecg{a} - \vecg{b} - 2\pi \vecg{n}|}  - \frac 1{|\vecg{b} - 2\pi\vecg{n}|}\right],
\ee
(we have set $g = L = 1$).
We see a nontrivial crystal structure in 6-dimensional configuration space.

The vector potentials  $\vecg{\cal A}^{(a,b)}$ are related to ${\cal D}_{(a,b)}$  so that
\be
\lb{cond-nab}
{\cal H}^{(a)}_j &\equiv& \varepsilon_{jkl}\, \pd_k^a {\cal A}^{(a)}_l \ =\ \pd_j^a {\cal D}_{(a)}\,, \nn
{\cal H}^{(b)}_j &\equiv& \varepsilon_{jkl} \,\pd_k^b {\cal A}^{(b)}_l \ =\ \pd_j^b {\cal D}_{(b)}\,, \nn
\varepsilon_{jkl} {\cal H}^{(ab)}_j  \ &\equiv& \ \pd_k^a {\cal A}^{(b)}_l  - \pd_l^b {\cal A}^{(a)}_k  \ =\   \varepsilon_{jkl}  \pd_j^a {\cal D}_{(b)}  \ =\   \varepsilon_{jkl}  \pd_j^b {\cal D}_{(a)}  \,.
 \ee
In Ref. \cite{Blok}, we performed an attempt to calculate the Witten index for this Hamiltonian and hence for the original theory. Unfortunately, we could not do that: the system is rather complicated, and it is difficult to solve the Schr\"odinger equation or the equations
$ \hat Q_\alpha^{\rm eff}  \Psi = \hat {\bar Q}^{\alpha \, {\rm eff} }\Psi = 0$.
analytically. An attempt to calculate the index via the functional integral in the approximation where only the zero Fourier modes in Euclidean time are taken into account  \cite{Cecotti} also failed: such a calculation gave a fractional result for the index. The reason is that the Hamiltonian \p{H-chiral-nab} involves strong singularities at $\vecg{a}_{\vecind{n}} = \vecg{0}$. These singularities invalidate the Cecotti-Girardello approximation: an estimate shows that the contributions of zero and nonzero Fourier modes in the functional integral are of the same order.\footnote{However, in spite of the fact that the CG approximation is not justified also in the Abelian case, such a calculation gave a correct answer \p{IW-Ab} for the Abelian theory! \cite{ja-Ab} The reason by which it works for the Hamiltonian \p{H-chiral}, but does not work for the Hamiltonian \p{H-chiral-nab} is  presently not clear.} 

In this paper, we will concentrate on studying the properties of the bosonic part of the Hamiltonian \p{H-chiral-nab} for $SU(3)$ and other groups. We will keep there only the kinetic part,
\be
\lb{H-Cartan}
\hat H =  \frac 23 \left[(\hat P_k^{(a)} + {\cal A}^{(a)}_k )^2 + (\hat P_k^{(b)} + {\cal A}^{(b)}_k )^2 + (\hat P_k^{(a)} + 
{\cal A}^{(a)}_k ) (\hat P_k^{(b)} + {\cal A}^{(b)}_k)  \right],
\ee
 and consider it only in the vicinity  of the ``corner" $\vecg{a}_0 = \vecg{b}_0 = \vecg{0}$ so that\footnote{We have introduced here a factor $q$ --- the ``charge" of the Cartan monopole. We will see soon that this charge can acquire any integer or half-integer value, as it was the case for the ordinary monopole. The actual value of the charge for the effective monopole of the chiral $SU(3)$ theory sitting in the corner is $q=-4$. } 
\be
\lb{corner}
{\cal D}_{(a)} = \ q \left[ \frac 1{\rho_{ab}}   - \frac 1{\rho_a} \right] , \quad 
{\cal D}_{(b)} =  q \left[  \frac 1{\rho_b} - \frac 1{\rho_{ab}}\right] 
\ee
with $\rho_a = |\vecg{a}|, \rho_b = |\vecg{b}|, \rho_{ab} = |\vecg{a} - \vecg{b}|$,
while the generalized vector potentials and magnetic fields are still related to ${\cal D}_{(a),(b)}$  by \p{cond-nab}. The magnetic fields have the form
\be
\lb{cornerH}
\vecg{\cal H}^{(a)} \ &=&\ q \left( \frac {\vecg{a}}{\rho_a^3} - \frac {\vecg{a} - \vecg{b}}{\rho_{ab}^3} \right)\,, \nn
\vecg{\cal H}^{(b)} \ &=&\ q \left( - \frac {\vecg{b}}{\rho_b^3}  +\frac {\vecg{b} - \vecg{a}}{\rho_{ab}^3} \right)\,, \nn
\vecg{\cal H}^{(ab)} \ &=&\  \frac {q(\vecg{a} - \vecg{b})}{\rho_{ab}^3}\,.
\ee
The vector potentials may be chosen in the gauge
\be
\lb{cornerA}
{\cal A}_x^{(a)} \ &=&\ q \left( -\frac {a_y}{\rho_a(\rho_a + a_z)} + \frac {(a-b)_y}{\rho_{ab}[\rho_{ab} + (a-b)_z]} \right) \,, \nn
{\cal A}_y^{(a)} \ &=&\ q\left(  \frac {a_x}{\rho_a(\rho_a + a_z)} - \frac {(a-b)_x}{\rho_{ab}[\rho_{ab} + (a-b)_z]} \right)\,, \nn
{\cal A}_x^{(b)} \ &=&\ q \left( -\frac {b_y}{\rho_b(\rho_b - b_z)} + \frac {(b-a)_y}{\rho_{ab}[\rho_{ab} - (b-a)_z]} \right) \,, \nn
{\cal A}_y^{(b)} \ &=&\ q \left( \frac {b_x}{\rho_b(\rho_b - b_z)} - \frac {(b-a)_x}{\rho_{ab}[\rho_{ab} - (b-a)_z]}  \right)\,, \nn
{\cal A}_z^{(a)} \ &=&\ {\cal A}_z^{(b)} \ =\ 0\,.
\ee
The symmetry relations
\be
\lb{sym-AH}
&&\vecg{\cal A}^{(a)} (-\vecg{b},-\vecg{a}) \ =\ - \vecg{\cal A}^{(b)}(\vecg{a},\vecg{b})\,, \nn
  &&\vecg{\cal H}^{(a)}(-\vecg{b},-\vecg{a}) \ =\  \vecg{\cal H}^{(b)}(\vecg{a},\vecg{b}), \qquad \vecg{\cal H}^{(ab)}(-\vecg{b},-\vecg{a}) \ =\  \vecg{\cal H}^{(ab)}(\vecg{a},\vecg{b})
 \ee
hold.
Note that the directions of the Dirac strings for ${\cal A}^{(a)}$ and ${\cal A}^{(b)}$ are correlated or, better to say, {\it anticorrelated}. Otherwise the symmetry \p{sym-AH} would not be there and the relations \p{cond-nab} would not be satisfied.

{\it This} is  what we will call a {\it Cartan monopole}.

\subsection{Dynamics of the system \p{H-Cartan}}
The Lagrangian that corresponds to \p{H-Cartan} reads
 \be
\lb{L-Cartan}
L \ =\ \frac 12 (\dot{\vecg{a}}^2 + \dot{\vecg{b}}^2 - \dot{\vecg{a}} \cdot  \dot{\vecg{b}}) - \dot{\vecg{a}} \vecg{\cal A}^a -
\dot{\vecg{b}} \vecg{\cal A}^b \,.
\ee

\begin{prop}
The action $\int L dt$ for the Lagrangian \p{L-Cartan} is  invariant under rotations
\be
\lb{rot-nab}
\delta a_j = \varepsilon_{jmn} \alpha_m a_n, \qquad \delta b_j = \varepsilon_{jmn} \alpha_m b_n\,.
 \ee
The corresponding integral of motion reads
\be
\lb{Mom-SU3}
J_m \ =\ 
\varepsilon_{mnk} \left[ a_n (P^{(a)} + {\cal A}^{(a)})_k  + b_n (P^{(b)} + {\cal A}^{(b)})_k \right] +  a_m {\cal D}_{(a)} +  b_m {\cal D}_{(b)}\,. 
\ee
\end{prop}

\begin{proof}
It goes along the similar lines as in the Abelian case.\footnote{Well, the Hamiltonian \p{H-Cartan} is defined on the Cartan subalgebra of $SU(3)$, which is also Abelian, but it came from non-Abelian studies.}

Consider the variation of the Lagrangian \p{L-Cartan} under \p{rot-nab}. We have
\be
\delta L \ =\ -\delta \left(\dot{a}_j {\cal A}_j^{(a)} + \dot{b}_j {\cal A}_j^{(b)} \right)  =\ -\alpha_m  \varepsilon_{jmn} \left(\dot a_n {\cal A}^{(a)}_j +
 \dot b_n {\cal A}^{(b)}_j  \right) \nn 
- \,  \alpha_m  \varepsilon_{pmn} \left[  \dot{a}_j \left(  \frac {\pd {\cal A}_j^{(a)}}{\pd a_p}  a_n  + \frac {\pd {\cal A}_j^{(a)}}{\pd b_p}  b_n   \right)    
+  \dot{b}_j \left(  \frac {\pd {\cal A}_j^{(b)}}{\pd a_p}  a_n  + \frac {\pd {\cal A}_j^{(b)}}{\pd b_p}  b_n   \right)   
 \right].
 \ee
Flipping the derivatives in the first term, we derive
\be
\lb{delta-L}
\delta L\ =\ - \alpha_m  \varepsilon_{jmn} \frac d{dt} \left(a_n {\cal A}^{(a)}_j +
  b_n {\cal A}^{(b)}_j  \right) \nn
+ \,\alpha_m  \varepsilon_{pmn} \left[a_n \dot a_j \left(\frac {\pd {\cal A}_p^{(a)}}{\pd a_j} -  \frac {\pd {\cal A}_j^{(a)}}{\pd a_p}  \right) 
+ b_n \dot b_j \left(\frac {\pd {\cal A}_p^{(b)}}{\pd b_j} -  \frac {\pd {\cal A}_j^{(b)}}{\pd b_p}  \right) \right]\nn
+ \,\alpha_m  \varepsilon_{pmn} \left[a_n \dot b_j \left(\frac {\pd {\cal A}_p^{(a)}}{\pd b_j} -  \frac {\pd {\cal A}_j^{(b)}}{\pd a_p}  \right) 
+ b_n \dot a_j \left(\frac {\pd {\cal A}_p^{(b)}}{\pd a_j} -  \frac {\pd {\cal A}_j^{(a)}}{\pd b_p}  \right) \right].
 \ee 
Use now the relations \p{cond-nab}. 
The second and the third line  in Eq.\p{delta-L} acquire the form  
\be
\lb{23-line}
\alpha_m \varepsilon_{pmn} \varepsilon_{jps} \left[a_n \dot a_j {\cal H}_s^{(a)} + b_n \dot b_j  {\cal H}_s^{(a)} +
( a_n \dot b_j + b_n \dot a_j) {\cal H}_s^{(ab)} \right]
\ee
Bearing in mind \p{cornerH}, the magnetic fields ${\cal H}_s^{(a)}$ and ${\cal H}_s^{(b)}$ can be represented as
$$  {\cal H}_s^{(a)}  \ =\ \frac {q a_s}{\rho_a^3} - {\cal H}_s^{ab}, \qquad  {\cal H}_s^{(b)}  \ =\  -\frac {q b_s}{\rho_a^3} - {\cal H}_s^{ab}\,. $$
Take the term $qa_s/\rho_a^3$ in ${\cal H}_s^{(a)}$ and substitute it in \p{23-line}. We may observe that this contribution in $\delta L$ has the same structure as in the Abelian case  [see Eq. \p{tot-der}] and boils down to a total derivative $\propto d(a_m/\rho_a)/dt$. By the same token, the contribution due to the term $-qb_s/\rho_b^3$ in ${\cal H}_s^{(b)}$ boils down to a total derivative $\propto d(b_m/\rho_b)/dt$. 
We are left with the contribution
\be
\alpha_m \varepsilon_{pmn} \varepsilon_{jps}\left[- a_n \dot a_j - b_n \dot b_j  + a_n \dot b_j + b_n \dot a_j \right] {\cal H}_s^{ab} \ = \nn
q\alpha_m(\delta_{mj} \delta_{ns} - \delta_{ms} \delta_{nj})(\dot a_j - \dot b_j)(a-b)_n(a-b)_s/\rho_{ab}^3\,,
\ee
which also gives a total derivative $\propto d[(a-b)_m/\rho_{ab}]/dt$.

The integral of motion \p{Mom-SU3} can be derived by the standard Noether procedure.
\end{proof}

In addition, the action is invariant under the transformations (the remnant of gauge transformations of the original field theory)
\be
\lb{gauge-nab}
{\cal A}_j^{(a)} \ \to \ {\cal A}_j^{(a)} - \pd_j^a F(\vecg{a}, \vecg{b}), \qquad {\cal A}_j^{(b)} \ \to \ {\cal A}_j^{(b)} - \pd_j^b F(\vecg{a}, \vecg{b})\,.
\ee
 The curls  $\vecg{\cal H}^{(a)} $, 
$\vecg{\cal H}^b $ and $\vecg{\cal H}^{(ab)} $  keep their form under \p{gauge-nab}. 

As is the case for the ordinary Dirac monopole, the Cartan monopole charge $q$ is quantized to be integer or  half-integer. To see this, consider the gauge transformation \p{gauge-nab} with
\be
F(\vecg{a}, \vecg{b}) \ =\ 2q \left[ \arctan \frac {a_y}{a_x}   + \arctan \frac {b_y}{b_x}
- \arctan \frac {(a-b)_y}{(a-b)_x} \right]\,.
\ee
After this transformation, the vector potentials \p{cornerA} acquire the form
 \be
\lb{cornerA1}
\tilde {\cal A}_x^{(a)} \ &=&\ q \left( \frac {a_y}{\rho_a(\rho_a - a_z)} - \frac {(a-b)_y}{\rho_{ab}[\rho_{ab} - (a-b)_z]} \right) \,, \nn
\tilde{\cal A}_y^{(a)} \ &=&\ q\left(  -\frac {a_x}{\rho_a(\rho_a - a_z)} + \frac {(a-b)_x}{\rho_{ab}[\rho_{ab} - (a-b)_z]} \right)\,, \nn
\tilde{\cal A}_x^{(b)} \ &=&\ q \left( \frac {b_y}{\rho_b(\rho_b + b_z)} - \frac {(b-a)_y}{\rho_{ab}[\rho_{ab} + (b-a)_z]} \right) \,, \nn
\tilde{\cal A}_y^{(b)} \ &=&\ q \left( -\frac {b_x}{\rho_b(\rho_b + b_z)} + \frac {(b-a)_x}{\rho_{ab}[\rho_{ab} + (b-a)_z]}  \right)\,, \nn
\tilde{\cal A}_z^{(a)} \ &=&\ \tilde{\cal A}_z^{(b)} \ =\ 0\,,
\ee
which corresponds to the opposite direction of the Dirac strings, compared to \p{cornerA}.

The eigenfunctions of \p{H-Cartan} with the vector potentials in \p{cornerA1} and in \p{cornerA} are related as
\be
\tilde \Psi(\vecg{a}, \vecg{b}) \ =\ \exp\{i F(\vecg{a}, \vecg{b})\} \Psi(\vecg{a}, \vecg{b}) \,.
\ee
For the wave function to be uniquely defined, we must require for $ \exp\{i F(\vecg{a}, \vecg{b})\}$ to be uniquely defined. And this is so iff $q$ is integer or  half-integer.


Note that one also can consider a generalized {\it asymmetric} Cartan monopole with the gauge potentials
\be
\lb{asymA}
{\cal A}_x^{(a)} \ &=&\   -q_a\frac {a_y}{\rho_a(\rho_a + a_z)} + q_{ab}\frac {(a-b)_y}{\rho_{ab}[\rho_{ab} + (a-b)_z]}  \,, \nn
{\cal A}_y^{(a)} \ &=&\ q_a \frac {a_x}{\rho_a(\rho_a + a_z)} - q_{ab}\frac {(a-b)_x}{\rho_{ab}[\rho_{ab} + (a-b)_z]} \,, \nn
{\cal A}_x^{(b)} \ &=&\ -q_b \frac {b_y}{\rho_b(\rho_b - b_z)} + q_{ab}\frac {(b-a)_y}{\rho_{ab}[\rho_{ab} - (b-a)_z]}  \,, \nn
{\cal A}_y^{(b)} \ &=&\ q_b\frac {b_x}{\rho_b(\rho_b - b_z)} - q_{ab}\frac {(b-a)_x}{\rho_{ab}[\rho_{ab} - (b-a)_z]} \,, \nn
{\cal A}_z^{(a)} \ &=&\ {\cal A}_z^{(b)} \ =\ 0\,.
\ee
or the potential related to \p{asymA} by \p{gauge-nab}. In particular, the gauge transformations 
\be
F_a(\vecg{a}, \vecg{b}) = 2 q_a \arctan \frac {a_y}{a_x}, \quad
F_b(\vecg{a}, \vecg{b}) = 2 q_b \arctan \frac {b_y}{b_x}, \quad
F_{ab}(\vecg{a}, \vecg{b}) = 2 q_{ab} \arctan \frac {(a-b)_y}{(a-b)_x}
\ee
reverse the direction of one of the Dirac strings in \p{asymA}. The condition for $\exp\{iF(\vecg{a}, \vecg{b})\}$ to be uniquely defined leads to the requirement for $q_a, q_b$ and $q_{ab}$ to be  integer or half-integer.

\subsection{Theory including a decuplet}

As was mentioned above, there are many different chiral anomaly-free theories. For example, the sextet $\Phi^{(jk)}$ in \p{chiral-SU3} may be replaced by a decuplet $\Phi^{(jkl)}$. The generators $T^a$ in this representation normalized so that $[T^a, T^b] = if^{abc} T^c$ read 
\be
\lb{Ta-10}
(T^a)_{jkl}^{pmn} \ =\ \frac 16 \left[ (t^a)_i^p (\delta_j^m \delta_k^n + \delta_j^n \delta_k^m) + {\rm 8\ similar\ terms} \right].
\ee
One can derive  (see Appendix) that 
 \be
\lb{27}
 {\rm Tr} \{T^{(a} T^b T^{c)}\} \ =\  27 \, {\rm Tr} \{t^{(a} t^b t^{c)} \},
\ee
and one needs 27 triplets to compensate the anomaly of an antidecuplet. 

The effective Hamiltonian in this theory may be calculated in the same way as in the theory \p{chiral-SU3}. We will not describe here the details of this calculation referring the reader to Refs. \cite{Blok,Wit-book} and to the Appendix and will just quote the result.
To assure the algebra \p{algN2}, the effective supercharges  and the Hamiltonian must have the same structure \p{Q-chiral},  \p{H-chiral-nab} as 
for the theory \p{chiral-SU3}, and they do! 

Remarkably, also the functions $D_{(a), (b)}$  have a quite similar structure:
\be
\lb{Dab-10}
\!\!\!\!\!{\cal D}_{(a)} = \ \frac {27}2 \sum_{\vecind{n}}\left[\frac 1{|\vecg{a} - 4\pi \vecg{n}|} - \frac 1{|\vecg{a} - \vecg{b} - 4\pi \vecg{n}|}   \right] \ - \  \frac 12     \sum_{\vecind{n}} \left[\frac 1{\left|\vecg{a} - \frac{4\pi \vecind{n}}3 \right|} - \frac 1{\left|\vecg{a} - \vecg{b} - \frac {4\pi\vecind{n}}3\right|} \right] , \nn
\!\!\!\!\!{\cal D}_{(b)} = \ \frac {27}2 \sum_{\vecind{n}}\left[ \frac 1{|\vecg{a} - \vecg{b} - 4\pi \vecg{n}|}  - \frac 1{|\vecg{b} - 4\pi \vecg{n}|}\right]\ - \ \frac 12 \sum_{\vecind{n}}\left[ \frac 1{\left|\vecg{a} - \vecg{b} - \frac{4\pi\vecind{n}}3 \right|}  - \frac 1{\left|\vecg{b} - \frac{4\pi\vecind{n}}3\right|}\right].
\ee

\section{Other groups and other monopoles}

\subsection{$SU(4)$}
Consider a chiral supersymmetric  theory based on $SU(4)$ gauge group, which involves eight matter multiplets in fundamental representation [the quartets with the Dynkin labels  (1,0,0)] and also an antidecuplet $\Phi^{(jk)}$ with the Dynkin labels (0,0,2). This is a direct $SU(4)$ analog of the theory \p{chiral-SU3}. 

We need 8 quartets for one antidecuplet to cancel the anomaly. Indeed, a generalization of the relation \p{dabc-SU3} to higher $SU(N)$ groups
is
\be
\lb{dabc-SUN}
 {\rm Tr} \{ T^{(a} T^b T^{c)} \} \ = \ (N+4) \,{\rm Tr} \{t^{(a} t^b t^{c)} \}. 
\ee
Take a Cartan background 
\be
\lb{Cartan-SU4}
\vecg{A}^a t^a \ =\ \frac 12 \,{\rm diag} (\vecg{a}, \vecg{b} - \vecg{a}, \vecg{c} - \vecg{b}, - \vecg{c})
\ee

The calculation of the effective supercharges and Hamiltonian goes along the same lines as for $SU(3)$. Some details of this calculation are given in the Appendix. We derive:
\be
\lb{Q-chiral-SU4}
   \hat Q_\alpha^{\rm eff} \ &=& \  {\eta^{(a)}_\beta} \left[(\sigma_k)_\alpha{}^\beta (\hat P_k^{(a)} + {\cal A}^{(a)}_k) 
+i  \delta_\alpha^\beta {\cal D}_{(a)}    \right] \nn
&+& \   {\eta^{(b)}_\beta} \left[(\sigma_k)_\alpha{}^\beta (\hat P_k^{(b)} + {\cal A}^{(b)}_k) +i  \delta_\alpha^\beta {\cal D}_{(b)}    \right],  \nn
&+&  {\eta^{(c)}_\beta}\left[(\sigma_k)_\alpha{}^\beta (\hat P_k^{(c)} + {\cal A}^{(c)}_k) 
+i  \delta_\alpha^\beta {\cal D}_{(c)}    \right]
 \ee 
and similarly for $ \hat {\bar Q}^{\alpha\, {\rm eff}}$ expressed via $\bar \eta^{(a)\alpha}$ etc. 
Here
\be
\eta_\alpha^{(a)} \ &=&\ \eta_\alpha^3 + \frac {\eta_\alpha^8}{\sqrt{3}} + \frac {\eta_\alpha^{15}}{\sqrt{6}} \,, \nn
\eta_\alpha^{(b)} \ &=&\  2\left(\frac {\eta_\alpha^8}{\sqrt{3}} + \frac {\eta_\alpha^{15}}{\sqrt{6}}  \right)\,, \nn
\eta_\alpha^{(c)} \ &=&\  \frac {3\eta_\alpha^{15}}{\sqrt{6}}\,.
\ee
Note that
\be
\lb{anticom-eta-4}
&& \{\bar \eta^{(a)\beta}, \eta^{(a)}_\alpha\} = \{\bar \eta^{(c)\beta}, \eta^{(c)}_\alpha\} = \frac 32 \delta_\alpha^\beta,
 \quad \{\bar \eta^{(b)\beta}, \eta^{(b)}_\alpha\}  = 2 \delta_\alpha^\beta\,, \nn
 &&\{\bar \eta^{(a)\beta}, \eta^{(b)}_\alpha\} = \{\bar \eta^{(b)\beta}, \eta^{(a)}_\alpha\} = \{\bar \eta^{(b)\beta}, \eta^{(c)}_\alpha\} = \{\bar \eta^{(c)\beta}, \eta^{(b)}_\alpha\} = \delta_\alpha^\beta, \nn
&& \{\bar \eta^{(a)\beta}, \eta^{(c)}_\alpha\} = \{\bar \eta^{(c)\beta}, \eta^{(a)}_\alpha\} =  \frac 12\delta_\alpha^\beta\,.
 \ee
The functions ${\cal D}_{(a)}$, ${\cal D}_{(b)}$ and ${\cal D}_{(c)}$ represent the following sums:
 \be
\lb{Dabc}
\!\!\!\!\!\!\!&&{\cal D}_{(a)} \ =\ 4 \sum_{\vecind{n}}  \left(\frac 1{|\vecg{a} - 4\pi \vecg{n} |}  -  \frac 1{|\vecg{a} - \vecg{b} - 4\pi \vecg{n}|}   \right)
- \frac 12 \sum_{\vecind{n}}  \left(\frac 1{|\vecg{a} - 2\pi \vecg{n} |}  -  \frac 1{|\vecg{a} - \vecg{b} - 2\pi \vecg{n} |}   \right) , \nn
\!\!\!\!\!\!\ &&{\cal D}_{(b)} \ =\ 4 \sum_{\vecind{n}}  \left(\frac 1{|\vecg{a} - \vecg{b} - 4\pi \vecg{n} |} -  \frac 1{|\vecg{b} - \vecg{c}  - 4\pi \vecg{n} |}   \right) - \frac 12 \sum_{\vecind{n}}  \left(\frac 1{|\vecg{a} - \vecg{b} - 2\pi \vecg{n} |} -  \frac 1{|\vecg{b} - \vecg{c}  - 2\pi \vecg{n} |}   \right), \nn
\!\!\!\!\!\!\ &&{\cal D}_{(c)} \ =\ 4 \sum_{\vecind{n}}  \left(\frac 1{|\vecg{b}  - \vecg{c} - 4\pi \vecg{n} |} -  \frac 1{|\vecg{c}  - 4\pi \vecg{n} |}   \right) - \frac 12\sum_{\vecind{n}}  \left(\frac 1{|\vecg{b}  - \vecg{c} - 2\pi \vecg{n} |} -  \frac 1{|\vecg{c}  - 2\pi \vecg{n} |}   \right).
\ee
The curls of ${\cal A}^{(a)}_k, {\cal A}^{(b)}_k, {\cal A}^{(c)}_k$ are related to the gradients of ${\cal D}_{(a)},  {\cal D}_{(b)}, {\cal D}_{(c)}$ in the same way as in \p{cond-nab}:
\be
\lb{cond-nab-4}
{\cal H}^{(a)}_j &\equiv& \varepsilon_{jkl}\, \pd_k^a {\cal A}^{(a)}_l \ =\ \pd_j^a {\cal D}_{(a)}\,, \nn
{\cal H}^{(b)}_j &\equiv& \varepsilon_{jkl} \,\pd_k^b {\cal A}^{(b)}_l \ =\ \pd_j^b {\cal D}_{(b)}\,, \nn
{\cal H}^{(c)}_j &\equiv& \varepsilon_{jkl} \,\pd_k^c {\cal A}^{(c)}_l \ =\ \pd_j^c {\cal D}_{(c)}\,, \nn
\varepsilon_{jkl} {\cal H}^{(ab)}_j  \ &\equiv& \ \pd_k^a {\cal A}^{(b)}_l  - \pd_l^b {\cal A}^{(a)}_k  \ =\   \varepsilon_{jkl}  \pd_j^a {\cal D}_{(b)}  \ =\   \varepsilon_{jkl}  \pd_j^b {\cal D}_{(a)} \,, \nn
\varepsilon_{jkl} {\cal H}^{(bc)}_j  \ &\equiv& \ \pd_k^b {\cal A}^{(c)}_l  - \pd_l^c {\cal A}^{(b)}_k  \ =\   \varepsilon_{jkl}  \pd_j^b {\cal D}_{(c)}  \ =\   \varepsilon_{jkl}  \pd_j^c {\cal D}_{(b)}  \,, \nn
\varepsilon_{jkl} {\cal H}^{(ac)}_j  \ &\equiv& \ \pd_k^a {\cal A}^{(c)}_l  - \pd_l^c {\cal A}^{(a)}_k  \ =\   \varepsilon_{jkl}  \pd_j^a {\cal D}_{(c)}  \ =\   \varepsilon_{jkl}  \pd_j^c {\cal D}_{(a)} \ =\ 0  \,.
 \ee

Calculating the anticommutator $\{\hat {\bar Q}^{\alpha\, {\rm eff}}, \hat Q_\alpha^{\rm eff} \}$, one can derive an expression for the effective Hamiltonian. It represents a rather obvious generalization of   \p{H-chiral-nab}. For example, bearing in mind \p{anticom-eta-4}, one can write a contribution
  \be 
\hat H^{\rm eff}_{\cal D}  \ = \ \frac 32
\left({\cal D}_{(a)}^2 +  {\cal D}_{(c)}^2\right) +  2 {\cal D}_{(b)}^2  + 
 {\cal D}_{(a)}{\cal D}_{(b)} + {\cal D}_{(b)}{\cal D}_{(c)}  +  \frac 12 {\cal D}_{(a)} {\cal D}_{(c)} 
 \ee
Another contribution represents an analogous quadratic form of  $\hat P_k^{(a,b,c)} + {\cal A}^{(a,b,c)}_k$:
\be 
\lb{only-A}
&&\hat H^{\rm eff}_{\cal A}  \ = \ \frac 32 \left[ (\hat P_k^a + {\cal A}_k^{(a)})^2 + (\hat P_k^c + {\cal A}_k^{(c)})^2 \right] +
2 (\hat P_k^b + {\cal A}_k^{(b)})^2 \nn
&&+ \ (\hat P_k^a + {\cal A}_k^{(a)})(\hat P_k^b + {\cal A}_k^{(b)})  +  (\hat P_k^b + {\cal A}_k^{(b)})(\hat P_k^c + {\cal A}_k^{(c)})
+ \frac 12 (\hat P_k^a + {\cal A}_k^{(a)})(\hat P_k^c + {\cal A}_k^{(c)}).
 \ee

 And there are 6 bifermion terms including $\vecg{\cal H}^{(a)},
 \vecg{\cal H}^{(bc)}$, etc.

The Cartan monopole for the $su(4)$ algebra may be defined as a system with the Hamiltonian \p{only-A} where  the vector potentials are related to ${\cal D}_{(a)}$, ${\cal D}_{(b)}$ and ${\cal D}_{(c)}$,  given by 
 \be
\lb{Dabc-Cartan}
&&{\cal D}_{(a)} \ =\ q  \left(\frac 1{|\vecg{a} - \vecg{b} |}   - \frac 1{|\vecg{a}  |}     \right), \nn
&&{\cal D}_{(b)} \ =\ q  \left(\frac 1{|\vecg{b} - \vecg{c}   |} - \frac 1{|\vecg{a} - \vecg{b}  |}     \right) , \nn
&&{\cal D}_{(c)} \ =\  q\left( \frac 1{|\vecg{c} |}   - \frac 1{|\vecg{b}  - \vecg{c}|}   \right),
\ee
as in \p{cond-nab-4}. One can choose
\be
\lb{cornerA-abc}
{\cal A}_x^{(a)} \ &=&\ q \left( -\frac {a_y}{\rho_a(\rho_a + a_z)} + \frac {(a-b)_y}{\rho_{ab}[\rho_{ab} + (a-b)_z]} \right) \,, \nn
{\cal A}_y^{(a)} \ &=&\ q\left(  \frac {a_x}{\rho_a(\rho_a + a_z)} - \frac {(a-b)_x}{\rho_{ab}[\rho_{ab} + (a-b)_z]} \right)\,, \nn
{\cal A}_x^{(b)} \ &=&\ q \left( \frac {(c-b)_y}{\rho_{bc}(\rho_{bc} + (c- b)_z)} + \frac {(b-a)_y}{\rho_{ab}[\rho_{ab} + (a-b)_z]} \right) \,, \nn
{\cal A}_y^{(b)} \ &=&\ q \left( \frac {(b-c)_x}{\rho_{bc}(\rho_{bc} +(c-b)_z)} + \frac {(b-a)_x}{\rho_{ab}[\rho_{ab} +  (a-b)_z]}  \right)\,, \nn
{\cal A}_x^{(c)} \ &=&\ q \left( \frac {(b-c)_y}{\rho_{bc}(\rho_{bc} + (c- b)_z)} + \frac {c_y}{\rho_c[\rho_c +c_z]} \right) \,, \nn
{\cal A}_y^{(c)} \ &=&\ q \left( \frac {(c-b)_x}{\rho_{bc}(\rho_{bc} +(c-b)_z)} - \frac {c_x}{\rho_{ab}[\rho_c  +c_z]}  \right)\,, \nn
{\cal A}_z^{(a)} \ &=&\ {\cal A}_z^{(b)}\  =  \  {\cal A}_z^{(c)} \ =\ 0\,.
\ee

 The corresponding Lagrangian reads
 \be
L \ =\ \frac {11}{60}(\dot {\vecg{a}}^2 + \dot {\vecg{c}}^2) + \frac 7{48} \dot {\vecg{b}}^2 - \frac 1{12} \dot {\vecg{b}} \cdot
(\dot {\vecg{a}} + \dot {\vecg{c}}) - \frac 1{30} \dot {\vecg{a}} \cdot \dot {\vecg{c}} \nn 
- \dot {\vecg{a}} \cdot \vecg{\cal A}^{(a)} - \dot {\vecg{b}} \cdot \vecg{\cal A}^{(b)} - \dot {\vecg{c}} \cdot \vecg{\cal A}^{(c)} \,.
 \ee
It has the properties similar to the Lagrangian \p{L-Cartan}: it is invariant under the gauge transformations,
\be
\lb{gauge-nab-SU4}
\!\!\!\!\!\!{\cal A}_j^{(a)} \to  {\cal A}_j^{(a)} + \pd_j^a F(\vecg{a}, \vecg{b}, \vecg{c}), \quad {\cal A}_j^{(b)}  \to  {\cal A}_j^{(b)} + \pd_j^b F(\vecg{a}, \vecg{b},  \vecg{c}), \quad {\cal A}_j^{(c)}  \to  {\cal A}_j^{(c)} + \pd_j^c F(\vecg{a}, \vecg{b},  \vecg{c}),
\ee
and under the rotations,\footnote{This can be proven in exactly the same way as in the $su(3)$ case.}
\be
\lb{rot-nab3}
\delta a_j = \varepsilon_{jmn} \alpha_m a_n, \qquad \delta b_j = \varepsilon_{jmn} \alpha_m b_n, \qquad \delta c_j = \varepsilon_{jmn} \alpha_m c_n\,.
 \ee
The conserved angular momentum reads:
\be
\lb{Mom-SU4}
J_m \ &=&\ 
\varepsilon_{mnk} \left[ a_n (P^{(a)} + {\cal A}^{(a)})_k  + b_n (P^{(b)} + {\cal A}^{(b)})_k  +  c_n (P^{(c)} + {\cal A}^{(c)})_k \right]\nn 
&+&  \, a_m {\cal D}_{(a)} + \, b_m {\cal D}_{(b)} + \,c_m {\cal D}_{(c)}\,. 
\ee
The charge $q$ is quantized to be integer or half-integer, as it was the case for the ordinary monopole and for the $su(3)$ Cartan monopole. 

\subsection{$SU(5)$}
For any $SU(N)$ group, one can consider the chiral supersymmetric gauge theory with $N+4$ fundamental multiplets carrying the Dynkin labels $(1,0,\ldots,0)$ and a symmetric multiplet with the labels $(0,\ldots,0,2)$. The chiral anomaly cancels for this theory. The effective Hamiltonians and the related Cartan monopoles have the form similar to that in the cases $N=3,4$ treated above.

But there is another interesting (especially, for possible phenomenological applications) anomaly-free chiral theory, which  involves a quintet with the Dynkin labels $(1,0,0,0)$ and an antisymmetric decuplet with the Dynkin labels $(0,0,1,0)$. 
The effective supercharges and Hamiltonian were determined in \cite{Blok}. Here we only quote the results. 
The slow bosonic variables have the form
  \be
\lb{Cartan-SU5}
\vecg{A}^a t^a \ =\ \frac 12 \,{\rm diag} (\vecg{a}, \vecg{b} - \vecg{a}, \vecg{c} - \vecg{b}, \vecg{d}- \vecg{c}, - \vecg{d})
 \ee
The effective supercharge looks similar to \p{Q-chiral-SU4}, but it now includes an extra term $\propto \eta_\alpha^{(d)}$.

The functions ${\cal D}_{(a)},  {\cal D}_{(b)}, {\cal D}_{(c)}, {\cal D}_{(d)}$ are given by the following sums
\be
&&\!\!\!\!{\cal D}_{(a)}  = \frac 12 \sum_{\vecind{n}}\left[  \frac 1{|\vecg{a}_{\vecind{n}}|} + \frac 1 {|(\vecg{a} - \vecg{c})_{\vecind{n}|} } + \frac 1 
{|(\vecg{b}- \vecg{a} + \vecg{d}- \vecg{c}) _{\vecind{n}} |} +  \frac 1 {|(\vecg{b} - \vecg{a} - \vecg{d})_{\vecind{n}} | }  \right. \nn
&&\left.  - \frac 1{|(\vecg{a} - \vecg{b})_{\vecind{n}}|} - \frac 1 {|(\vecg{a} + \vecg{c} - \vecg{b})_{\vecind{n}}| } - \frac 1 {|(\vecg{a} + \vecg{d} - \vecg{c})_{\vecind{n}} |} - \frac 1 {|(\vecg{a} - \vecg{d})_{\vecind{n}} | } \right],\nn
&&\!\!\!\!{\cal D}_{(b)}  = \frac 12 \sum_{\vecind{n}}\left[  \frac 1{|(\vecg{a}- \vecg{b})_{\vecind{n}}|} + \frac 1 {|(\vecg{b} - \vecg{d})_{\vecind{n}|} } + \frac 1 
{|(\vecg{a}+ \vecg{c} - \vecg{b}) _{\vecind{n}} |} +  \frac 1 {|(\vecg{b} - \vecg{c} + \vecg{d})_{\vecind{n}} | }  \right. \nn
&&\left.  - \frac 1{|\vecg{b}_{\vecind{n}}|} - \frac 1 {|(\vecg{b} - \vecg{c})_{\vecind{n}}| } - \frac 1 {|(\vecg{b} - \vecg{a} - \vecg{d})_{\vecind{n}} |} - \frac 1 {|(\vecg{b} - \vecg{a} + \vecg{d} - \vecg{c})_{\vecind{n}} | } \right],\nn
&&\!\!\!\!{\cal D}_{(c)}  = \frac 12 \sum_{\vecind{n}}\left[  \frac 1{|\vecg{c}_{\vecind{n}}|} + \frac 1 {|(\vecg{b} - \vecg{c})_{\vecind{n}|} } + \frac 1 
{|(\vecg{a}- \vecg{c} + \vecg{d}) _{\vecind{n}} |} +  \frac 1 {|(\vecg{b} - \vecg{a} + \vecg{d} - \vecg{c})_{\vecind{n}} | }  \right. \nn
&&\left.  - \frac 1{|(\vecg{d} - \vecg{c})_{\vecind{n}}|} - \frac 1 {|(\vecg{a} - \vecg{c})_{\vecind{n}}| } - \frac 1 {|(\vecg{a} + \vecg{c} - \vecg{b})_{\vecind{n}} |} - \frac 1 {|(\vecg{c} - \vecg{b}- \vecg{d})_{\vecind{n}} | } \right],\nn
&&\!\!\!\!{\cal D}_{(d)}  = \frac 12 \sum_{\vecind{n}}\left[  \frac 1{|(\vecg{d}- \vecg{c})_{\vecind{n}}|} + \frac 1 {|(\vecg{d} - \vecg{a})_{\vecind{n}|} } + \frac 1 
{|(\vecg{c}- \vecg{b} - \vecg{d}) _{\vecind{n}} |} +  \frac 1 {|(\vecg{b} - \vecg{a} - \vecg{d})_{\vecind{n}} | }  \right. \nn
&&\left.  - \frac 1{|\vecg{d} _{\vecind{n}}|} - \frac 1 {|(\vecg{b}  - \vecg{d})_{\vecind{n}}| } - \frac 1 {|(\vecg{a} + \vecg{d} - \vecg{c})_{\vecind{n}} |} - \frac 1 {|(\vecg{b} - \vecg{a} +\vecg{d} - \vecg{c})_{\vecind{n}} | } \right],
\ee
wihere $\vecg{a}_{\vecind{n}} = \vecg{a} - 4\pi \vecg{n}$. In this case, the elementary cell of our 12-dimensional crystal includes only one node: the antisymmetric decuplet has only mixed components, which have the same periodicity pattern as the quintet components. The anticommutator of the supercharges gives the Hamiltonian, which has a structure analogous to the structure of $SU(3)$ and $SU(4)$ effective  Hamiltonians, but  has more terms.  For example, the scalar potential is now given by the following quadratic form:
 \be
&\hat H^{\rm eff}_{\cal D}  = \frac 15 
\left[4\left({\cal D}_{(a)}^2 +  {\cal D}_{(d)}^2\right) +  6\left({\cal D}_{(b)}^2 +  {\cal D}_{(c)}^2\right) + 6\left({\cal D}_{(a)} {\cal D}_{(b)} + {\cal D}_{(c)} {\cal D}_{(d)}\right)\right. \nn
& \left.+  \ 4\left({\cal D}_{(a)} {\cal D}_{(c)} + {\cal D}_{(b)} {\cal D}_{(d)}\right) + 2{\cal D}_{(a)} {\cal D}_{(d)} + 8{\cal D}_{(b)} {\cal D}_{(c)}\right]\,.
 \ee
It would be interesting to study the dynamics of the corresponding Cartan monopole.

\subsection{$Spin(4n+2)$ and $E_6$}

Similar constructions exist for all groups admitting complex representations: for all higher unitary groups, for $Spin(4n+2)$ and for $E_6$. The specific of $Spin(4n+2)$ and  $E_6$ theories is that they are free from anomalies for {\it any} matter content, and we can consider only the theories including a single matter supermultiplet. For example, a fundamental spinor 16-plet in $SO(10)$ \cite{Frich} or a fundamental 27-plet in $E_6$ \cite{Gursey}. The reason is that, though these representations are complex, the tensor $d^{abc} = {\rm Tr}\{T^{(a} T^b T^{c)}\} $ {\it vanishes} for these groups. 

The simplest way to see this is to inspect  Tables 44, 48 in the review \cite{Slansky}. The symmetrized tensor products of the adjoint representations in $SO(10)$ and 
$E_6$ are 
\be
SO(10): &\qquad&  ({\bf 45} \times {\bf 45} )_s  \ =\  {\bf 1} + {\bf 54} + {\bf 210} + {\bf 770}\,, \nn
E_6: &\qquad&  ({\bf 78} \times {\bf 78} )_s  \ =\  {\bf 1} + {\bf 650} + {\bf 2430}\,.
 \ee
Adjoint representations are absent in these expansions, which means that $T^{(a}T^b T^{c)}$ is not a singlet and has a zero trace.

 For orthogonal groups, one can give a simple mathematical reason for this fact \cite{Ilia}.

 Consider the form $F = {\rm Tr} \{X^3 \}$ defined on a Lie algebra $\mathfrak{g}$. It is invariant under the action of the group. Expanding $X = X_a t^a$, one can represent it as
 $$ F \ =\ \frac 14 d^{abc} X_a X_b X_c\,.$$
 Now we can  restrict this form to $X = X_{\tilde a}t^{\tilde a}$, $\tilde a = 1,\ldots,r$, belonging to the Cartan subalgebra of  $\mathfrak{g}$. Then $\tilde F = {\rm Tr} \{(X_{\tilde a} t^{\tilde a})^3\} \propto d^{\tilde a \tilde b \tilde c} X_{\tilde a} X_{\tilde b} X_{\tilde c}$  is invariant under the action of the Weyl group. 
For example, for  the Cartan subalgebra element $X = {\rm diag}(\alpha, \beta, \gamma = -\alpha - \beta)$ of $su(3)$, 
$\tilde F = \alpha^3 + \beta^3 + \gamma^3$, which is
  invariant under the permutation of $\alpha, \beta, \gamma$. Similarly, for higher unitary groups. 

The existence of a Weyl-invariant homogeneous cubic polynomial defined on the Cartan subalgebra is a necessary and sufficient condition for the trace  Tr$\{T^{(a} T^b T^{c)}\} $ to be nonzero.

Consider now the algebra $so(10)$ and its Cartan subalgebra. An element of the latter is a 5-dimensional vector $x_j$. The Weyl group includes the permutations of $x_j$ and also the reflections when a {\it couple} of components of  $x_j$ changes simultaneously the sign. But
a nonzero cubic polynomial of $x_j$ invariant under all the Weyl transformations does not exist. Indeed, 
to be invariant under the reflections that change the sign of any pair of $x_2, x_3, x_4, x_5$, but not of $x_1$, such a  polynomial should have the form
$P_3(x_j) = x_1 \sum_{j=2}^5 a_j x_j^2$. But such a polynomial is not invariant under permutations.  

\vspace{1mm}

A more heuristic argument in favor of the absence of the chiral anomaly  in $SO(10)$ and $E_6$ comes from the branching rules for their multiplets under embeddings $E_6 \supset SO(10) \times U(1)$ and $SO(10) \supset SU(5) \times U(1)$. The fundamental 27-plet of $E_6$ is split into a spinor, vector and singlet of $SO(10)$, ${\bf 27}  \rightarrow {\bf 16} + {\bf 10} + {\bf 1}$, and the 16-plet of $SO(10)$ is further split into a quintet, antidecuplet and singlet of $SU(5)$.
And we have already learned that the $SU(5)$ theory with such a matter content is anomaly-free.

The absence of anomalies in the theories based on $SO(10)$ and $E_6$ led people to consider these theories including each a single fundamental matter multiplet as candidats for the theories of grand unification
\cite{Frich,Gursey}. It would be  interesting to evaluate the effective Hamiltonians for the supersymmetric version of these theories in the same spirit as we did it for the unitary groups, to study them and find out in the end of the day whether supersymmetry is spontaneously broken there or not.

We are indebted to K. Konishi and I. Smilga for illuminating discussions.

\vspace{1cm}

\appendix

\section{Evaluation of effective Hamiltonians}
\renewcommand{\theequation}{A.\arabic{equation}}
\setcounter{equation}0

\subsection*{Decuplet contribution in $SU(3)$ theory}
Let us first derive \p{27}. One can do it quite directly, using the representation \p{Ta-10}, but there exist a simpler and nicer way \cite{Konishi}.
To find the coefficient in (3.24), it suffices to calculate Tr$\{(T^8)^3\}$ and compare it with Tr$\{(t^8)^3\}$.  It is convenient to choose the normalization
$t^8 = {\rm diag}(1,1,-2)$ and the corresponding normalization for $T^8$ in the higher representations.   Consider the embedding $SU(2) \times U(1) \subset SU(3)$ such that $T^8$ is projected on the generator of the $U(1)$ subgroup. The branching rules for the triplet, sextet and decuplet are:
  \be
{\bf 3}  \ &\rightarrow& \ {\bf 2} (1) + {\bf 1} (-2), \nn
{\bf 6}  \ &\rightarrow& \ {\bf 3} (2) +  {\bf 2}(-1) + {\bf 1} (-4), \nn
{\bf 10}  \ &\rightarrow& \ {\bf 4} (3) + {\bf 3}(0) + {\bf 2}(-3) + {\bf 1} (-6), 
 \ee  
where {\bf 2,3,4} are the $SU(2)$ multiplets and the numbers in the parentheses are the $U(1)$ charges.  Then Tr$\{(t^8)^3\} = 2\cdot 1^3 + 1\cdot (-2)^3 = -6$. For the sextet and decuplet we have
 \be
{\rm Tr}_{\bf 6} \{(T^8)^3 \} \ &=&\  3 \cdot 2^3 + 2 \cdot (-1)^3 + 1 \cdot (-4)^3  = -42 \ =\ 7  \, {\rm Tr}\{(t^8)^3\}, \nn
{\rm Tr}_{\bf 10} \{(T^8)^3 \}  \ &=& \  4 \cdot 3^3 + 2 \cdot (-3)^3 + 1 \cdot (-6)^3  = -162 \ =\ 27 \,  {\rm Tr}\{(t^8)^3\}.
\ee
For an arbitrary symmetric representation $(p,0)$,
\be
{\rm Tr}_p \{(T^8)^3 \} \ =\  \frac{p(p+1)(p+2)(p+3)(2p+3)}{120}\, {\rm Tr}\{(t^8)^3\}\,.
 \ee

\vspace{1mm}

As was mentioned, the effective supercharges and Hamiltonian have the form \p{Q-chiral}, \p{H-chiral-nab}, and the only problem is to evaluate the contribution of the matter multiplets in ${\cal D}_{(a)}$ and ${\cal D}_{(b)}$. The contribution of a triplet and a sextet was found in \cite{Blok}, and these calculations were spelled out in detail in \cite{Wit-book}, where we refer the reader. For a triplet, it was
\be
\lb{Dab-3}
\!\!\!\!\!{\cal D}_{(a)}({\rm tripl}) &=& \frac 12 \sum_{\vecind{n}}\left[\frac 1{|\vecg{a} - 4\pi \vecg{n}|} - \frac 1{|\vecg{a} - \vecg{b} - 4\pi \vecg{n}|}   \right]   , \nn
\!\!\!\!\!{\cal D}_{(b)}({\rm tripl}) &=& \frac 12 \sum_{\vecind{n}}\left[ \frac 1{|\vecg{a} - \vecg{b} - 4\pi \vecg{n}|}  - \frac 1{|\vecg{b} - 4\pi \vecg{n}|}\right]  \,,
\ee 
where the individual terms in these sums came from the matter Fourier modes $\propto e^{2\pi i \vecind{n} \cdot \vecind{x}/L}$. The antisextet contribution was 
\be
\lb{Dab-6}
\!\!\!\!\!{\cal D}_{(a)}({\rm sext}) =\  \sum_{\vecind{n}}\left[\frac 1{|\vecg{a} - 4\pi \vecg{n}|} - \frac 1{|\vecg{a} - \vecg{b} - 4\pi \vecg{n}|}   \right]  - \frac 12 \sum_{\vecind{n}}\left[\frac 1{|\vecg{a} - 2\pi \vecg{n}|} - \frac 1{|\vecg{a} - \vecg{b} - 2\pi \vecg{n}|}   \right]  , \nn
\!\!\!\!\!{\cal D}_{(b)} ({\rm sext}) = \ \sum_{\vecind{n}}\left[ \frac 1{|\vecg{a} - \vecg{b} - 4\pi \vecg{n}|}  - \frac 1{|\vecg{b} - 4\pi \vecg{n}|}\right]
 - \frac 12 \sum_{\vecind{n}}\left[\frac 1{|\vecg{a} - 2\pi \vecg{n}|} - \frac 1{|\vecg{a} - \vecg{b} - 2\pi \vecg{n}|}   \right]\,.
\ee 
Here the terms in the first sums come from the mixed sextet component $\phi^{12}, \phi^{13}, \phi^{23}$ while the terms in the second sums come from the components $\phi^{11}(\vecg{x}), \phi^{22}(\vecg{x}), \phi^{33}(\vecg{x})$, which, in contrast to the mixed components and the triplet fields, make ``two turns", being multiplied by $e^{\pm 4\pi i x}$ rather than by $e^{\pm 2\pi i x}$ to compensate the gauge shifts
 $a_x \to a_x+ 4\pi$ or   $b_x \to b_x + 4\pi $. Multiplying \p{Dab-3} by 7 and adding \p{Dab-6}, we arrive at \p{Dab}.

To determine the contribution of the decuplet, we first find the part of the Hamiltonian describing the interaction of the slow Abelian gauge fields \p{Cartan} with the fast scalar decuplet fields. 

Let us  first concentrate on the zero Fourier modes of the decuplet. The kinetic term of the fast Hamiltonian reads
 \be
\lb{Hkin}
\hat H^{\rm kin} \ &=&\ |\hat \pi^{111}|^2 +  |\hat \pi^{222}|^2 +  |\hat \pi^{333}|^2  \nn
&+& \, 3 \left(  |\hat \pi^{112}|^2 +  |\hat \pi^{122}|^2 +  |\hat \pi^{113}|^2 +  |\hat \pi^{133}|^2 +  |\hat \pi^{223}|^2 +  |\hat \pi^{233}|^2  \right) 
 + 6 |\hat \pi^{123}|^2\,,
 \ee
where $\hat \pi^{111} = -i \pd/\pd \phi^{111}$, etc.
After some calculation, we derive also the potential part:
\be
\lb{Hpot}
V \ &=&\ \frac 94 \left[\vecg{a}^2 | \phi^{111}|^2   + 
(\vecg{a} - \vecg{b})^2  | \phi^{222}|^2 + \vecg{b}^2  | \phi^{333}|^2  \right] + \frac 34 (\vecg{a} + \vecg{b})^2 \left[ | \phi^{112}|^2
+  | \phi^{332}|^2    \right] \nn
&+&\ \frac 34 (2\vecg{a} - \vecg{b})^2 \left[ | \phi^{113}|^2
+  | \phi^{223}|^2    \right] +  \frac 34 (2\vecg{b} - \vecg{a})^2 \left[ | \phi^{221}|^2
+  | \phi^{331}|^2    \right] \,.
 \ee
The same can be done for the Hamiltonian describing the interaction with nonzero Fourier modes. In the same way as in the sextet case, this Hamiltonian involves the mixed components like $\phi^{112}(\vecg{x}) $ making one turn to compensate the gauge shift $a_x \to a_x + 4\pi$, etc. and also the components
$\phi^{111}(\vecg{x}) $, $\phi^{222}(\vecg{x}) $, $\phi^{333}(\vecg{x}) $ making {\it three} turns.
As we see, no potential is generated along the direction $\phi_{123}$. (The existence of such direction is a specific of the decuplet problem.)  In nine other directions (for each $\vecg{n}$), 
the Hamiltonian $\hat H^{\rm kin} + V$ has an oscillator form. From that one can determine the vacuum averages: 
 \be
\lb{srednie}
 &&\langle|\phi^{111}_{\vecind{n}}|^2 \rangle  \ =\   \frac 1{3\left|\vecg{a} - \frac {4\pi}3 \vecg{n} \right|} , \qquad  \langle|\phi^{222}_{\vecind{n}}|^2 \rangle  \ =\  \frac 1{3\left|\vecg{a} - \vecg{b} - \frac {4\pi}3 \vecg{n} \right|} , \qquad \langle|\phi^{333}_{\vecind{n}}|^2 \rangle  \ =\   \frac 1{3\left|\vecg{b} - \frac {4\pi}3 \vecg{n} \right|} , \nn
&&\langle|\phi^{112}_{\vecind{n}}|^2 \rangle \ =\  \langle|\phi^{332}_{\vecind{n}}|^2\rangle   \ =\  \frac 1{\left|\vecg{a} +\vecg{b} -  4\pi \vecg{n} \right|} , \qquad 
\langle|\phi^{113}_{\vecind{n}}|^2\rangle \ =   \ \langle|\phi^{223}_{\vecind{n}}|^2\rangle  \  = \ \frac 1{\left|2\vecg{a} -\vecg{b} -  4\pi \vecg{n} \right|} , \qquad  \nn
&&\langle|\phi^{221}_{\vecind{n}}|^2\rangle \ =\  \langle|\phi^{331}_{\vecind{n}}|^2\rangle   \ =\  \frac 1{\left|2\vecg{b} - \vecg{a} -  4\pi \vecg{n} \right|}\,.
\ee
The effective supercharges can be derived by averaging the full supercharges $\hat Q_\alpha, \hat {\bar Q}^\alpha$ over the fast vacuum. We have 
\be
\lb{Q-full}
\hat Q_\alpha \ =\ i\int d\vecg{x} \left\{\eta_\gamma^3 \left[ -(\sigma_k)_\alpha{}^\gamma \frac \pd {\pd A^3_k} + \delta_\alpha^\gamma (s^*_f t^3 s_f -  \phi^* T^3 \phi) \right]   \right. \nn
+ \left. \eta_\gamma^8 \left[ -(\sigma_k)_\alpha{}^\gamma \frac \pd {\pd A^3_k} + \delta_\alpha^\gamma (s^*_f t^8 s_f  -   \phi^* T^8 \phi) \right]  \right\} \ + \ \hat Q_\alpha^{\rm fast} \,.
 \ee
Here $\hat Q_\alpha^{\rm fast} $ is the fast supercharge. It gives zero acting on the ground state of $\hat H^{\rm fast}$. Now, $s_{f = 1, \ldots, 27}$
are triplet scalar fields, $\phi^{jkl}$ are the decuplet fields and the generators $T^{3,8}$ were written in Eq. \p{Ta-10}. Plugging in \p{Q-full} the decuplet scalar averages \p{srednie}, the scalar averages of seven triplets,
\be
\langle |s_1^f|^2 \rangle \ =\ \frac 1{|\vecg{a} - 4\pi \vecg{n}|}, \qquad \langle |s^f_2|^2 \rangle \ =\ \frac 1{|\vecg{a} - \vecg{b} - 4\pi \vecg{n} |}, \qquad \langle |s^f_3|^2 \rangle \ =\ \frac 1{|\vecg{b} - 4\pi \vecg{n}| }  \,,
\ee
and adding the contributions from the averages of $\pd/\pd A_k^3$ and $\pd/\pd A_k^8$ over the full fast ground state wave function,\footnote{These contributions bring about the ``vector potentials" ${\cal A}_k^{(a)}$ and ${\cal A}_k^{(b)}$. }
we arrive at \p{Q-chiral} with $D_{(a), (b)}$ given in \p{Dab-10}.

Note that the effective supercharge does not acquire a contribution from the average $|\phi^{123}|^2$, which we could not determine in the above analysis.
Note also that, in contrast to what we had for the sextet, the contributions of the mixed components of $\phi^{jkl}$ with the same periodicity pattern as the triplet fields {\it cancel} so that the first sums in \p{Dab-10} are due solely to the triplets while the second sums are due solely to the decuplet. 

\subsection*{$SU(4)$ with eight quartets and a symmetric antidecuplet}

This theory was not treated before, and we give here a little more calculational details.

 Consider first the contribution of a quartet. The fast Hamiltonian describing the interaction of the  slow background \p{Cartan-SU4} with the zero Fourier modes of the scalar quartet components reads 
 \be
\hat H^{\rm fast}_{\rm quart} \ =\  \hat \pi^\dagger_j \hat \pi_j + s^* (\vecg{A}^a t^a )^2 s   \nn
= \ - \frac {\pd^2}{\pd s^{j*} \pd s_j} + \frac 14[ \vecg{a}^2 |s_1|^2 + (\vecg{a} - \vecg{b})^2 |s_2|^2 + (\vecg{b} - \vecg{c})^2 |s_3|^2
+ \vecg{c}^2 |s_4|^2\,.
\ee
This gives the averages  
\be 
\left \langle |s_1|^2  \right \rangle \ =\ \frac 1{|\vecg{a}|}, \quad \left \langle |s_2|^2  \right \rangle \ =\ \frac 1{|\vecg{a} - \vecg{b}|}, \quad 
\left \langle |s_3|^2  \right \rangle \ =\ \frac 1{|\vecg{b} - \vecg{c}|}, \quad \left \langle |s_4|^2  \right \rangle \ =\ \frac 1{|\vecg{c}|}\,.
\ee
The averages of the other Fourier modes are given by the similar expressions with $\vecg{a} \to \vecg{a} - 4\pi \vecg{n}$ etc.

The  contribution of a quartet to the effective supercharge $\hat Q_\alpha^{\rm eff}$ is
\be
 i \eta^3_\alpha \, s^* t^3 s +  i \eta^8_\alpha \, s^* t^8 s +  i \eta^{15}_\alpha \, s^* t^{15} s \ =\ i\left(\eta_\alpha^{(a)}
 {\cal D}^{\rm quart}_{(a)} +   \eta_\alpha^{(b)} {\cal D}^{\rm quart}_{(b)}   +\eta_\alpha^{(c)} {\cal D}^{\rm quart}_{(c)}\right)
\ee
with 
\be
\lb{D-quart}
{\cal D}^{\rm quart}_{(a)} \  &=& \ \frac 12 \sum_{\vecind{n}} \left(\frac 1{|\vecg{a} - 4\pi \vecg{n}|} -  \frac 1{|\vecg{a} - \vecg{b} - 4\pi\vecg{n}|} \right), \nn 
{\cal D}^{\rm quart}_{(b)}  \ &=&\  \frac 12 \left(\frac 1{|\vecg{a} - \vecg{b}- 4\pi \vecg{n}|} -  \frac 1{|\vecg{b} - \vecg{c}- 4\pi \vecg{n}|} \right), \nn
{\cal D}^{\rm quart}_{(c)} \ &=& \ \frac 12 \left(\frac 1{|\vecg{b} - \vecg{c}- 4\pi \vecg{n}|} -  \frac 1{|\vecg{c}- 4\pi \vecg{n}|} \right).
\ee  

The fast Hamiltonian for the scalar antidecuplet (its zero Fourier mode) interacting with the Cartan background is 
\be
\lb{}
\hat H^{\rm fast}_{\rm dec} \ =\ \hat \pi^\dagger_{ij} \pi^\dagger_{ij} + \vecg{a}^2 |\phi_{11}|^2 + (\vecg{a} - \vecg{b})^2 |\phi_{22}|^2
+ (\vecg{b} - \vecg{c})^2 |\phi_{33}|^2 +  \vecg{c}^2  |\phi_{44}|^2 \nn
+ \frac 12\left[ \vecg{b}^2 (|\phi_{12}|^2 + |\phi_{34}|^2) + (\vecg{a} + \vecg{c} - \vecg{b})^2  (|\phi_{13}|^2 + |\phi_{24}|^2) +  (\vecg{a} - \vecg{c})^2  |\phi_{14}|^2 + |\phi_{23}|^2) \right],
\ee
giving the averages
\be
|\phi_{11}|^2 = \frac 1{2|\vecg{a}|}, \quad |\phi_{22}|^2 = \frac 1{2|\vecg{a} - \vecg{b}|}, \quad 
|\phi_{33}|^2 = \frac 1{2|\vecg{b} - \vecg{c}|}, \quad |\phi_{44}|^2 = \frac 1{2|\vecg{c}|}, \nn
|\phi_{12}|^2 =  |\phi_{34}|^2  =  \frac 1{|\vecg{b}|}, \quad   |\phi_{23}|^2 =  |\phi_{14}|^2  =  \frac 1{|\vecg{a} - \vecg{c}|}, \quad
|\phi_{13}|^2 =  |\phi_{24}|^2  =  \frac 1{|\vecg{b} - \vecg{a} - \vecg{c}|}\,.      
\ee
The induced ${\cal D}_{(a,b,c)}$ are

\be
\lb{D-dec4}
{\cal D}^{\rm zero\ antidecuplet\ modes}_{(a)} \  &=& |\phi_{22}|^2  - |\phi_{11}|^2  \ = \  \frac 12 \left(\frac 1{|\vecg{a} - \vecg{b}|} - \frac 1{|\vecg{a}|} \right),\nn
{\cal D}^{\rm zero\ antidecuplet\ modes}_{(b)} \  &=& |\phi_{33}|^2  - |\phi_{22}|^2  \ = \  \frac 12 \left(\frac 1{|\vecg{b} - \vecg{c}|} - \frac 1{|\vecg{a} - \vecg{b}|} \right),\nn
{\cal D}^{\rm zero\ antidecuplet\ modes}_{(c)} \  &=& |\phi_{44}|^2  - |\phi_{33}|^2  \ = \  \frac 12 \left(\frac 1{|\vecg{c}|} - \frac 1{|\vecg{b} - \vecg{c}|}\right),
\ee
while the contributions of the mixed components exactly cancel!

The contributions of the other Fourier modes have the same form with $\vecg{a} \to \vecg{a}  - 2\pi \vecg{n}$ etc. (the components like $\phi^{11}(\vecg{x})$ rotate 2 times faster than the quartet components).  

Adding the antidecuplet contribution to \p{D-quart} taken with the factor 8, we arrive at 
\p{Dabc}.

\end{document}